\documentclass[a4paper,UKenglish,cleveref, autoref, thm-restate]{lipics-v2019}

\title{Parameterized String Equations}

\author{Laurent Bulteau}{LIGM, CNRS, Université Gustave Eiffel, Marne-la-Vallée France}{laurent.bulteau@univ-eiffel.fr}{https://orcid.org/0000-0003-1645-9345}{}

\author{Michael R. Fellows}{Department of Informatics
University of Bergen
Bergen, Norway
}{michael.fellows@uib.no}{}{}
\author{Christian Komusiewicz}{Fachbereich für Mathematik und Informatik, Philipps-Universität Marburg, Germany}{komusiewicz@informatik.uni-marburg.de}{https://orcid.org/0000-0003-0829-7032}{}
\author{Frances Rosamond}{Department of Informatics
University of Bergen
Bergen, Norway}{frances.rosamond@uib.no}{}{}

\authorrunning{L. Bulteau, M.\,R. Fellow, C. Komusiewicz and F.\,A. Rosamond}

\Copyright{Laurent Bulteau, Michael R. Fellows,\\ Christian Komusiewicz, Frances Rosamond}



\category{} 

\relatedversion{} 

\supplement{}

\acknowledgements{We want to thank Shonan Workshop 045 in which this work was initiated}

\nolinenumbers 

\hideLIPIcs  

\EventEditors{John Q. Open and Joan R. Access}
\EventNoEds{2}
\EventLongTitle{42nd Conference on Very Important Topics (CVIT 2016)}
\EventShortTitle{ARXIV 2021}
\EventAcronym{CVIT}
\EventYear{2016}
\EventDate{December 24--27, 2016}
\EventLocation{Little Whinging, United Kingdom}
\EventLogo{}
\SeriesVolume{42}
\ArticleNo{X}

\usepackage{amsmath}%
\usepackage{amsfonts}%
\usepackage{amssymb,amsthm}%
\usepackage{graphicx}
\usepackage{xspace}
\usepackage{cleveref}
\usepackage{tikz}
\usepackage{todonotes}
\usepackage{bm}
\usepackage{dsfont}
\usepackage{mathtools}

\newcommand{\splitsIn}{\equiv}
\newcommand{\joker}{{*}}
\newcommand{\XP}{{\sf XP}\xspace}
\newcommand{\NP}{{\sf NP}\xspace}
\newcommand{\FPT}{{\sf FPT}\xspace}
\newcommand{\Wone}{{\sf W[1]}\xspace}
\newcommand{\Wtwo}{{\sf W[2]}\xspace}
\newcommand{\Pclass}{{\sf P}\xspace}
\newcommand{\Blocks}{{\mathcal B}}
\newcommand{\Ecal}{{\mathcal E}}
\newcommand{\Oh}{{\mathcal O}}

\newcommand{\Tgt}{T}
\newcommand{\Tgtlen}{t}
\newcommand{\Src}{{\mathcal X}}
\newcommand{\src}{X}

\newcommand{\leftpart}{\text{\sf pre}}
\newcommand{\rightpart}{\text{\sf suf}}

\newcommand{\dunderline}[3][-4pt]{{%
  \sbox0{#3}%
  \ooalign{\copy0\cr\rule[\dimexpr#1-#2\relax]{\wd0}{#2}}}}
  
\newcommand{\colorunderline}[3][-4pt]{{\color{#2}\dunderline[#1]{1.5pt}{{\color{black} \ensuremath{#3}}}}}

\newcommand{\ur}[1]{\colorunderline{red!80!white}{#1}}
\newcommand{\ub}[1]{\colorunderline{blue!30!black}{#1}}
\newcommand{\ug}[1]{\colorunderline{green}{#1}}
\newcommand{\ugd}[1]{\colorunderline[-6pt]{green}{#1}}

\begin{document}

\maketitle

\begin{abstract}
We study systems of \emph{String Equations} where block variables need to be assigned strings so that their concatenation gives a specified target string. We investigate this problem under a  multivariate complexity framework, searching for tractable special cases  such as systems of equations with few block variables or few equations. Our main results include a polynomial-time algorithm for size-2 equations, and hardness for size-3 equations, as well as hardness for systems of two equations, even with tight constraints on the block variables. 
We also study a variant where few deletions are allowed in the target string, and give XP algorithms in this setting when the number of block variables is constant.

\end{abstract}

\section{Introduction}

\emph{String equations} are equations in which the variables, called \emph{block variables} in this work, can be assigned arbitrary strings. We study the case where a string equation can express that the concatenation of a given list of blocks must equal a specific string. We study the problem of finding a solution for a given set of string equations within a multivariate complexity framework, searching for tractable special cases  such as systems of equations with few blocks variables or few equations.

\paragraph*{Related Work}
The single-equation problem can be seen as a \emph{String Morphism} question, for which many variants have been extensively studied in the literature (see Fernau et al.~\cite{FSV16} for an overview of related works). A possible starting point is the \emph{pattern discovery} problem \cite{angluin1979finding}: given a set of strings, find a pattern (i.e. a string of characters and variables), of minimal length, that may generate each input string by assigning values to each block variable. Deciding if a given pattern generates a string can be seen, in our setting, as a single string equation.  Angluin \cite[Theorem~3.2.3]{angluin1979finding} proves that this problem is \NP-hard, by reduction from SAT. She also identifies the brute-force $O(n^k)$ algorithm (which requires linear space~\cite{IBARRA1995179}). Fernau et al.~\cite{FS15, FSV16} further prove (among other variants) \Wone-hardness when parameterized by the number of block variables. A notable variant is the \emph{injective} case (where blocks must be assigned distinct strings) which in its simplest version can be formulated as the following \NP-hard problem~\cite{FMMS20}: given a string $T$ and an integer $k$, can $T$ be split into $k$ distinct factors?  

For two equations, if the blocks are used once in each equation, the {\sc String Equations} formulation can be related to {\sc Minimum Common String Partition} (where two target strings are given, but the block variables can be permuted in any order). Interestingly, this problem is \FPT for the number of blocks~\cite{BK14} but becomes \Wone-hard if the permutation of the blocks is fixed (Theorem~\ref{prop:W-hard-2eq-emptyBlocks}).

\paragraph*{Formal definitions}

Given a string $S$, we write $|S|$ for the length of $S$, and $S[i]$ for the $i$th letter of $S$ (for $1\leq i\leq |S|$). We use $\prod_{i=1}^{n} S_i:= S_1 S_2 \ldots S_n$ to denote the concatenation of strings~$S_1$ to~$S_n$. 
We distinguish \emph{substrings} (obtained by taking consecutive characters from $s$) from \emph{subsequences} (obtained by taking not-necessarily consecutive characters from $s$).

Given an alphabet $\Sigma$ and a set $\Blocks$ of \emph{block variables} (also called \emph{blocks} for simplification), with $\Sigma$ and $\Blocks$ being disjoint, 
a \emph{string equation} over $(\Sigma,\Blocks)$ is a pair $(\Tgt,\Src)$, written $\Tgt\splitsIn \Src$, where $\Tgt$ (resp. $\Src$) is a non-empty string over $\Sigma$ (resp. $\Blocks$). $\Tgt$ is the \emph{target string} of the equation. 
A \emph{system of equations} is  a set of string equations.
A block which appears only once in a given system of equations is said to be a \emph{joker}, and is denoted $\joker$. A system of equations is \emph{duplicate-free} if no block appears twice in the same string $\Src$.

An \emph{assignment} of the block variables is a function $\sigma:\Blocks\rightarrow \Sigma^+$ that assigns to each block variable a non-empty string over $\Sigma$. 
For $\Src=\src_1\src_2\ldots \src_c\in \Blocks^+$, we write $\sigma(\Src)=\prod_{i=1}^c\sigma(\src_i)$.
An equation $\Tgt\splitsIn \Src$ is \emph{satisfied} by an assignment $\sigma$ if $\Tgt=\sigma(\Src)$, that is, replacing $\Src$ by the concatenation of the corresponding strings yields exactly $\Tgt$.
A system of equations is satisfied if all equations are satisfied.
An example is given in \Cref{fig:basic_example}.

\begin{figure}

\begin{align*}
 \ur{a\,b\,c}\,\ug{a\,b} &\splitsIn \ur A\,\ug B       
 &
 \sigma(A)&:= \ur{a\,b\,c} \\
 \ur{a\,b\,c}\,\ub{d}\,\ur{a\,b\,c}\,\ub{d} & \splitsIn \ur A\,\ub C\,\ur A\,\ub C 
 &
 \sigma(B)&:= \ug{a\,b} \\
 \ug{a\,b}\,\ub{d} &\splitsIn \ug B\,\ub C         
 &
 \sigma(C)&:= \ub{d}
\end{align*}
    \caption{\label{fig:basic_example} Left: a string equation system with alphabet $\Sigma=\{a,b,c,d\}$ and block variables $\Blocks=\{A,B,C\}$. Right: a solution (assignment) $\sigma$ for this system, with the corresponding substrings  highlighted with different colors in the target strings.
}

\end{figure}

\paragraph*{The String Equations Problem and its Variants}
We define the {\sc String Equations} problem as follows: 

\noindent{\bf Input:}  A system of equations $\Ecal=\{\Tgt_1\splitsIn \Src_1, \Tgt_2\splitsIn \Src_2, \ldots, \Tgt_r\splitsIn \Src_r\}$.
\\
\noindent{\bf Question:} Does $\Ecal$ admit a satisfying assignment $\sigma$?

We consider the \emph{duplicate-free} restriction of {\sc String Equations}, where no block may appear twice in any $\Src_i$. We also consider a more general problem where a number of deletions is allowed, {\sc String Equations with Deletions}.  

\noindent{\bf Input:}  A system of equations $\Ecal=\{\Tgt_1\splitsIn \Src_1, \Tgt_2\splitsIn \Src_2, \ldots, \Tgt_r\splitsIn \Src_r\}$, an integer $d$\\
\noindent{\bf Question:} Does there exist 
$r$ strings $\Tgt_1',\ldots,\Tgt_r'$ obtained from $\Tgt_1,\ldots,\Tgt_r$ by removing a total of at most $d$ letters, such that the equations $\Tgt_i'\splitsIn \Src_i$ admit a satisfying
assignment $\sigma$?

Finally, motivated by the similarities with the {\sc Minimum Common String Partition} problem, we are interested in the restriction with $r=2$ duplicate-free equations without joker blocks. In this setting, we consider the variation where blocks may be assigned empty strings (remember that by default, $\sigma$ must be \emph{non-erasing}, i.e. $\sigma(\Tgt)$ is not empty), see \Cref{prop:W-hard-2eq-emptyBlocks}.

Since the size of each $\sigma(A)$ can be bounded by the input size, and it is trivial to verify whether an equation is satisfied by a given assignment, {\sc String Equations} is clearly in {\sf NP}. Similarly, {\sc String Equations with Deletions} and all variants mentioned here are also clearly in {\sf NP}.

\paragraph*{Parameterized Complexity}
For an introduction to parameterized complexity theory, we refer to the monograph of Downey and Fellows~\cite{DowneyF13}.
An instance~$(x,\kappa)$ of a parameterized problem consists of a classical problem instance~$x$ and a parameter~$\kappa \in \mathds{N}$. A parameterized problem is \emph{fixed-parameter tractable}. or equivalently is in \FPT, if there exists an algorithm that solves every instance~$(x,\kappa)$ in~$f(\kappa)\cdot |x|^{\Oh(1)}$~time. A parameterized problem has an \XP-algorithm, or equivalently is in \XP, if there exists an algorithm that solves every instance~$(x,\kappa)$ in~$|x|^{g(\kappa)}$~time. To show that problems in \XP are unlikely to be in \FPT, one may use parameterized reductions from \Wone-hard problems, i.e., polynomial-time reduction between parameterized problems such that the parameter of the reduced instance is a function of the orginal parameter.

We use the following classic \Wone-hard problem in our reduction.  

\noindent {\sc Clique}\\
\noindent{\bf Input:}  A graph $G=(V,E)$\\
\noindent{\bf Question:} Does~$G$ contain a clique of size $\kappa$?

In addition, we use a constrained version of the problem, where the vertices are colored so that a clique must use exactly one vertex with each color :

\noindent {\sc Multicolored Clique}\\
\noindent{\bf Input:}  A $\kappa$-partite graph $G=(V=\bigcup_{i=1}^{\kappa} V_i,E)$. 
\\
\noindent{\bf Question:} Does~$G$ contain a clique of size $\kappa$ in $G$?

{\sc Multicolored Clique} is \Wone-hard for parameter $\kappa$~\cite{FHRV09,Pietrzak03}. An example instance is given in Figure~\ref{fig:graph}, it will be used as an example for our reductions. By convention we write $n=|V|$, $m=|E|$, $V=\{v_1,\ldots,v_n\}$, $E=\{e_1,\ldots, e_m\}$. For convenience, we assume that edges are size-2 strings over alphabet $V$, that is, $e_j=v_s v_t$ for an edge $\{v_s,v_t\}$ with $s<t$. Moreover, we assume that $E$ is sorted lexicographically, that is, for $e_j=v_s v_t$ and $e_{j'}=v_{s'}v_{t'}$ we have $j<j'$ iff $s<s'$ or $s=s'$ and $t<t'$. 

\begin{figure} 
  \centering
  \begin{tikzpicture}[yscale=0.6, node distance= 1em, draw, circle,inner sep=4pt, line width=0pt]
    \node[fill=red!20] (a) at ( 110:1) {};   \node at (a) {$a$};    
    \node[fill=green!15] (b) at (210:1) {};   \node at (b) {$c$};    
    \node[fill=blue!40] (c) at (330:1) {};   \node at (c) {$d$};    
    \node[fill=red!20] (d) at (70:1) {};   \node at (d) {$b$};    
    \draw (a) -- (b);
    \draw (a) -- (c);
    \draw (b) -- (c);
    \draw (b) -- (d);
    \draw (c) -- (d);
  \end{tikzpicture}
  \caption{\label{fig:graph} Running example for reductions from {\sc Clique} (or {\sc Multi-colored Clique}): a (3-colored) graph $G=(V,E)$ where $V=\{a,b,c,d\}$ and $E=\{ab,ac,bc,bd,cd \}$. For $\kappa=3$, $G$ contains two triangles: $\{a,b,c\}$ and $\{b,c,d\}$. 
  }
\end{figure}
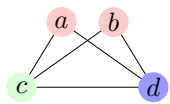

\section{Overview of the results}

In our complexity study, we consider the following five quantities: 
\begin{itemize}
    \item $\Tgtlen$, the maximum size of any target string 
    \item $r$, the number of equations
    \item $c$, the maximum number of blocks per equation,
    \item $k$, the overall number of blocks
    \item $d$, the number of deletions (for {\sc String Equations with Deletions})    
\end{itemize}
We only only consider $r$, $c$, $k$ and $d$ as possible parameters (i.e. we always assume $\Tgtlen$ to be unbounded). Observe that $k\leq rc$. Consequently, a positive result for parameter~$k$ implies a positive result for~$rc$ and a running time lower bound for~$r+c$ implies a running time lower bound for~$k$.  An overview of our results is given in Table~\ref{tab:results}.

We first show in Section~\ref{sec:constant_k} that the problem is polynomial-time solvable when $k$ is any constant (i.e. \XP for $k$, \Cref{prop:XP}) using brute-force enumeration. This raises the possibility of an \FPT algorithm for parameter $k$, which will be ruled out in the next sections, with hardness results proving \Wone-hardness for parameter $k$ even in more restricted settings.  
We also consider the deletions variant: our \XP algorithm extends to parameter $k+d$ (\Cref{prop:del+k-XP}), but with parameter $k$ alone, we show that it is already \NP-hard for $k=1$ (\Cref{prop:del-hard-k=1}). However, it is in \XP for parameter $r+c$ (\Cref{prop:del-rc-XP}).

We then focus, in Section~\ref{sec:constant_r}, on the number $r$ of equations. \Wone-hardness is already known for $r=1$ and parameter $c$~\cite{angluin1979finding, FSV16}. 
We prove a similar result in the duplicate-free setting when $r=2$ ($r=1$ is a trivial case for duplicate-free equations).  Note that \Wone-hardness follows for parameter $k$, since $k\in \Oh(c)$ for constant $r$. We consider two slightly stronger cases: one where both target strings are equal, and the other where block variables may be assigned empty strings.

Finally, Section~\ref{sec:constant_c} is devoted to small-size equations (constant $c$). On the negative side, {\sc String Equations} remains \Wone-hard for parameter $r$ (and $k$) when $c=3$. On the positive side, we give a polynomial-time algorithm for size-2 equations, which generalizes to the case where non-joker blocks must all be prefixes or suffixes of their equations. 
As a final attempt towards generalizing this algorithm to other cases, we consider equation systems with $(r-1)$ size-2 equations together with a single size-$c$ duplicate-free equation. We show, however, that this special case is \Wone-hard for parameter $r+c$ (\Cref{prop:W-hard-all-but-one-size2}).


\begin{table}
    \centering
    \begin{tabular}{l|ccccc|lr}    
   &$k$ & $r$& $c$ &$d$ & &\multicolumn{2}{l}{Complexity of {\sc String Equations} }
\\ \hline \multirow{6}{3cm}{Constant number of blocks (Section~\ref{sec:constant_k})}
&C& $\star$&  $\star$& 0 & & \XP &(\Cref{prop:XP})
\\ &C& $\star$&  $\star$& C &&  \XP& (\Cref{prop:del+k-XP})
\\ &C& $\star$&  $\star$& P &&  \multicolumn{2}{l}{{\em open} 
} 
\\ &1&  1&  P& $\star$ &&  \Wtwo-hard& (\Cref{prop:del-hard-k=1})
\\ &1&  P&  1& $\star$ &&  \Wtwo-hard&(\Cref{prop:del-hard-k=1})
\\ &C&  C&  C& $\star$ &&  \XP&(\Cref{prop:del-rc-XP})
 \\ \hline \multirow{2}{3cm}{Constant number of equations (Section~\ref{sec:constant_r})}
 & P&  1&  P& 0 &&  \Wone-hard & \cite{angluin1979finding, FSV16} (see also  \Cref{prop:W-hard-1eq})
\\ &P&  2&  P& 0 &$^\dagger$ &  \Wone-hard & (\Cref{prop:W-hard-2eq,prop:W-hard-2eq-emptyBlocks}) 
\\ \hline \multirow{4}{3cm}{Constant equation sizes (Section~\ref{sec:constant_c})}
&$\star$&  $\star$&  2& 0 &$^\dagger$  & \Pclass  &(\Cref{prop:P-border-blocks}) 
\\ & $\star$&  $\star$&  2& P &&  \multicolumn{2}{l}{{\em open} 
} 
\\ &P&  P&  3& 0 &&  \Wone-hard& (\Cref{prop:W-hard-size3})
\\ &P&  P&  P& 0 &$^\dagger$ &  \Wone-hard & (\Cref{prop:W-hard-all-but-one-size2}) 
    \end{tabular}
    \caption{Summary of our results . For each of $k$ (number of block variables), $r$ (number of equations), $c$ (equation size) and $d$ (number of deletions), we indicate if the result holds for the given integer, for any fixed constant (C, for \XP algorithms), when seen as a parameter (P, for \FPT or \Wone-hardness) or is unbounded ($\star$). \\$^\dagger$ These results still hold in a  stronger setting: \Cref{prop:W-hard-2eq} holds for duplicate-free equations  with a unique target. \Cref{prop:W-hard-2eq-emptyBlocks} holds for duplicate-free equations allowing empty blocks. \Cref{prop:P-border-blocks} holds for arbitrarily large equations, provided non-border blocks are all jokers. \Cref{prop:W-hard-all-but-one-size2} holds for duplicate-free equations, and all equations except one have size 2. }
    \label{tab:results}
\end{table}
\section{Constant number of blocks} \label{sec:constant_k}
A straightforward brute-force algorithm that guesses endpoints of each block shows that {\sc String Equations} is in \XP.
\begin{proposition} \label{prop:XP}
{\sc String Equations} can be solved in $\Oh^*(\Tgtlen^{2k})$ time.
\end{proposition}
\begin{proof}
For each block $X\in\Blocks$, pick an equation $\Tgt\splitsIn \Src$ where $X$ appears in $\Src$, and branch into all possibilities to choose integers $1\leq i \leq j\leq |\Tgt|\le  \Tgtlen$ such that $\sigma(X):=\Tgt[i]\ldots \Tgt[j]$. Overall, this creates $\Oh(\Tgtlen^{2k})$ branches. For each branch, it remains to check whether the resulting assignment satisfies $\Ecal$, which can be done in linear time. The total running time is thus $\Oh^*(\Tgtlen^{2k})$. 
\end{proof}
The algorithm above can be extended to the case where we allow for deletions in the target strings.
\begin{proposition} \label{prop:del+k-XP}
{\sc String Equations with Deletions} is in \XP for parameter $k+d$.
\end{proposition}
\begin{proof}
Branch into the at most
$\binom{tr}{d}$ 
ways of building strings $\Tgt_1',\ldots,  \Tgt_r'$ out of 
$\Tgt_1,\ldots , \Tgt_r$. 
Testing each time if the resulting {\sf String Equation} problem has a solution can be done via \Cref{prop:XP}, so the overall running time is $O^*((nr)^dt^{2k})$. 
\end{proof}

\begin{theorem} \label{prop:del-hard-k=1}
{\sc String Equations with Deletions} is \NP-hard for $k=1$. Moreover, {\sc String Equations with Deletions} is \Wtwo-hard for parameter $r$ when $k=c=1$ and \Wtwo-hard for parameter $c$ when $k=r=1$.
\end{theorem}
\begin{proof}
We note that {\sc String Equations with Deletions} contains {\sc Longest Common Subsequence} as a sub-problem, using a single block $X$ and equations $\Tgt_1\splitsIn X,\ldots, \Tgt_r\splitsIn X$ to denote the fact that $\sigma(X)$ should be a common subsequence of each $\Tgt_i$. 
In this setting, minimizing the number of deletions $d= \sum_{i=1}^r |\Tgt_i| -r|\sigma(X)|$ is equivalent to maximizing $|\sigma(X)|$, i.e. the length of the common subsequence. Since LCS is \Wtwo-hard for the parameter number of strings $r$~\cite{BDFW95},
the same result applies to  
{\sc String Equations with Deletions} for $k=c=1$.

For the case with a single equation ($r=1$), it suffices to insert a sufficiently long prefix $P=\$^d$ before each string, then again $P\Tgt_1P\Tgt_2\ldots P\Tgt_r \splitsIn XX\ldots X$ has a solution reaching the target number of deletions $d$ iff $\Tgt_1,\ldots, \Tgt_r$ have an LCS reaching the corresponding target size $\frac 1r(\sum_{i=1}^r |\Tgt_i| - d)$.
\end{proof}

\begin{proposition} \label{prop:del-rc-XP}
{\sc String Equations with Deletions} is in \XP for parameter $r+c$.
\end{proposition}
\begin{proof}
We introduce the following notation: given an equation $\Tgt\splitsIn \Src$ satisfied by an assignment $\sigma$, the \emph{starting point} for block $i$ is the index $j$, $1\leq j\leq |\Tgt|$, such that the first character of $\sigma(\Src[i])$ is mapped to $\Tgt[j]$. Thus, if $j$ is the starting point of block $i$ and $j'$ is the starting point of block $i+1$ (or $j'=|\Tgt|+1$ if $i=|\Src|$), then $\sigma(\Src[i])$ is a subsequence of $\Tgt[j]\ldots \Tgt[j'-1]$. 

Consider a solution $\sigma$ for an instance $(\Ecal, d)$ of {\sc String Equations with Deletions}. Compute the starting points of each block in each equation: by the above remark, for each $X\in\Blocks$ there exist substrings of $\Tgt_i$s denoted $F_1, \ldots, F_h$ such that $\sigma(X)$ is a common subsequence of strings $F_1, \ldots, F_h$ (with $h\leq rc$ corresponding to the number of occurrences of block $X$). Note that replacing $\sigma(X)$ by any longest common subsequence of these strings still yield a valid solution for the same problem, since the number of deletions performed in target strings does not increase.

This leads to the following \XP algorithm: branch into all possibilities to choose a starting point of every block of $\Ecal$; the number of branches is $\Oh(\Tgtlen^{rc})$. For each block $X$, compute strings $F_1,\ldots, F_h$ as described above. Set $\sigma(X)$ to be a longest common subsequence of these strings in time $O^*(\Tgtlen^h)$:  it remains to check whether $\sigma$ is a solution to the original problem, which can be done in linear time. 
\end{proof}

\section{Constant number of equations}\label{sec:constant_r}
For the sake of completeness (and as a warm-up for following, more complex reductions), we give a \Wone-hardness proof for the single-equation case (see also Fernau et al. \cite{FSV16}).
\begin{theorem} \label{prop:W-hard-1eq}
{\sc String Equations}  with one equation (that is, $r=1$) is W[1]-hard for~$c$.
\end{theorem}

\begin{proof}

\begin{figure}
\begin{align*}
\Tgt&= y_0\ \ur a\,\ug c\ y_1\ \ur a\,\ub d\ y_2\  b\, c\ y_3\ b\,d\ y_4\ \ug c\,\ub d\ y_5 
\\
\Tgt&\splitsIn \joker\ \ur R\,\ug G \ \joker\  \ur R\,\ub B \ \joker\  \ug G\,\ub B \ \joker
\end{align*}
\caption{\label{fig:W-hard-1eq} Reduction of \Cref{prop:W-hard-1eq} (using $r=1$ equation) applied to the graph from \Cref{fig:graph} (using a single string equation, with the target string described first, then its block decomposition). In this and following examples, in order to lighten notations, we use color initials $R,G,B$ or $r,g,b$ instead of indices $1,2, 3$ to highlight the fact that we select a vertex of the corresponding color in the graph (or an edge between corresponding colors, i.e. $R,G,B$ designate $X_1,X_2,X_3$ respectively, $R_g,R_b$ designate  $X_{1,2}, X_{1,3}$, etc). A satisfying assignment is given by colored underlines: $\sigma(R)=a$, $\sigma(G)=c$, $\sigma(B)=d$.}
\end{figure}

See \Cref{fig:W-hard-1eq}. Consider an instance $(G=(V,E), \kappa)$ of {\sc Clique}.  Introduce $m+1$ separators (i.e., new characters) denoted $y_0,\ldots, y_{m}$, let $\Sigma=V\cup\{y_0,\ldots,y_m\}$. Introduce $\kappa$ non-joker blocks $\{X_1,\ldots,X_{\kappa}\}$ and ${\kappa \choose 2}+1$ joker blocks (for a total of $k=\Oh( \kappa^2)$ blocks). Let 
$$\Ecal:=\{\Tgt\splitsIn \Src\} \text{ where } \Tgt:=y_0\prod_{i=1}^m (e_j y_j) \text{ and } \Src:= \joker\prod_{1\leq i < \kappa}\prod_{i< j \leq \kappa} (X_iX_j \joker)$$
This concludes the construction of the instance. We now show the correctness of the reduction. 

($\Rightarrow$) Assume that $G$ has a $\kappa$-clique $\{u_1,\ldots, u_{\kappa}\}$. Set $X_i:=u_i$. Then each $X_iX_j$ with $1\leq i<j\leq \kappa$ corresponds to an edge in $G$, hence it appears as a substring in $\Tgt$, in lexicographical order (recall that we assume that the edges of~$G$ are ordered). Finally, joker blocks can be matched to the gaps around selected edges, which are non-empty (they contain at least one separator), so this assignment satisfies $\Ecal$.

($\Leftarrow$) Assume that $\Ecal$ has a satisfying assignment $\sigma$. First, note that $\sigma(X_i)$ may contain only characters from $V$ (since $X_i$ is repeated $\kappa-1$ times and separators only once),  and at most one character (otherwise, $\sigma(X_iX_j)$ contains three characters, so by definition of $\Tgt$ at least one separator). Thus $X_i\in V$. Furthermore, $X_iX_j$ is an edge of $E$ for each $1\leq i<j\leq \kappa$, so $\{X_1,\ldots,X_{\kappa}\}$ is a $\kappa$-clique.
\end{proof}

Only with one equation the problem is already hard, but this definitely needs duplications (since an instance having 1 duplicate-free equation is trivial). Hence we look at instances with a constant number of equations, but no duplicates: we show this case is hard as well, even in the \emph{unique target} case, i.e. when all target strings are equal.

\begin{theorem} \label{prop:W-hard-2eq}
{\sc String Equations} with $r=2$ duplicate-free equations, a unique target and parameter $c$ is W[1]-hard.
\end{theorem}

\begin{figure}

\begin{align*}
\Tgt &=z\ x_0\ \ur a\,\ur a\ x_1\ b\,b\ x_2\ \ug c\,\ug c\ x_3\ \ub d\,\ub d\ x_4\ z'\ y_0\ \ur a\,\ug c\ y_1\ \ur a\,\ub d\ y_2\ b\,c\ y_3\ b\,d\ y_4\ \ug c\,\ub d\ y_5
\\
\Tgt&\splitsIn Z\ A_0\ \ur{R_g}\,\ur{R_b} \ A_1\  \ug{G_r}\,\ug{G_b} \ A_2\  \ub{B_r}\,\ub{B_g} \ A_3\  Z' \ B_0\ \ur{R_g'}\,\ug{G_r'} \ B_1\  \ur{R_b'}\,\ub{B_r'} \ B_2\  \ug{G_b'}\,\ub{B_g'} \ B_3
\\
\Tgt&\splitsIn Z'\ A_0\ \ur{R_g'}\,\ur{R_b'} \ A_1\  \ug{G_r'}\,\ug{G_b'} \ A_2\  \ub{B_r'}\,\ub{B_g'} \ A_3\  Z \ B_0\ \ur{R_g}\,\ug{G_r} \ B_1\  \ur{R_b}\,\ub{B_r} \ B_2\  \ug{G_b}\,\ub{B_g} \ B_3
\end{align*}
\caption{\label{fig:W-hard-2eq} Reduction of \Cref{prop:W-hard-2eq} (using $r=2$ duplicate-free equations) applied to the graph from \Cref{fig:graph}.
}
\end{figure}
\begin{proof}

See \Cref{fig:W-hard-2eq}. Consider an instance $G=(V,E), \kappa$ of {\sf Clique}.  Introduce $n+m+3$ separators (i.e., new characters) denoted $x_0,\ldots, x_{n}$, $y_0,\ldots, y_{m}$ and $z$. Let $\Sigma$ consist of $V$ and the set of separators. Introduce $2\kappa(\kappa-1)$ \emph{coding} blocks $\{X_{i,j},X_{i,j}'\}$ for $i\neq j$, a \emph{starting} block $Z$ and $\kappa+{\kappa \choose 2}+2$ \emph{gap} blocks $A_0,\ldots, A_{\kappa} , B_0$ and $B_{i,j}$ for $1\leq i<j\leq \kappa$. Define the following string equations (we decompose the strings into \emph{vertex} and \emph{edge} sections for ease of presentation):
\begin{align*} 
\Ecal&:=\{\Tgt\splitsIn \Src, \Tgt\splitsIn \Src'\} \\
\text{ where } \Tgt&:=z\quad x_0\ \prod_{i=1}^n (v_i^{k-1}\ x_i )
     &&z\quad y_0\ \prod_{j=1}^m (e_j\ y_j) \\
\Src&:= Z\ A_0\  \prod_{1\leq i \leq \kappa}\prod_{j\neq i}( X_{i,j}\ A_{i,j} ) 
     &&Z'\ B_0\  \prod_{1\leq i < \kappa}\prod_{i< j \leq \kappa}( X'_{i,j}X'_{j,i}\ B_{i,j} ) \\
\Src'&:=\underbrace{Z'\ A_0\  \prod_{1\leq i \leq \kappa}\prod_{j\neq i}(X'_{i,j}\ A_{i,j} ) }_{\text{vertex section}}
     &&\underbrace{Z\ B_0\   \prod_{1\leq i < \kappa}\prod_{i< j \leq \kappa} (X_{i,j}X_{j,i}\ B_{i,j} )}_{\text{edge section}} 
\end{align*}
This concludes the construction of the instance. 
We now show the correctness of the reduction.

($\Rightarrow$) If $G$ has a $\kappa$-clique $K=\{w_1,\ldots w_{\kappa}\}$. Set $\sigma(Z)=\sigma(Z')=z$ and $\sigma(X_{i,j})=\sigma(X'_{i,j})=w_i$ for all $j\neq i$. Then each $\sigma(\prod_{j\neq i} X_{i,j})$ equals $x_i^{k-1}$, and these substrings appear in the same  order as in $\Src$ in the vertex section of $\Tgt$, and $\sigma(X'_{i,j}X'_{j,i})$, with $i<j$, correspond to edges of $G$ which appear in this order in the edge section of $\Tgt$. Pick $\sigma(A_i)$ to be the gaps around the selected vertices and $\sigma(B_{i,j})$ to be the gaps around selected edges. We match $\Src'$ in exactly the same way with $\Tgt$ (exchanging the roles of $X_i$ and $X_i'$), so that the gaps are identical in both matchings. Thus this assignment satisfies $\Ecal$.

($\Leftarrow$) If $\Ecal$ has a satisfying assignment $\sigma$,
First, note that $\sigma(Z)=z$: The first block in $\Src$ is $Z$, so $Z$ must start with $z$. Moreover,~$z$ is followed by different characters in $\Src$ and $\Src'$ ($x_0$ and $y_0$, respectively). Similarly, $\sigma(Z')=z$. Thus, the vertex (edge) sections of $\Src$ and $\Src'$ are matched only to the vertex (edge) section of $\Tgt$. Hence, $\sigma(X_{i,j})$ is a substring of both  sections 
of~$\Tgt$, which may only be single characters corresponding to some character $w_{i,j}$. In the vertex section, characters $\sigma(X_{i,j})$  for $j\neq i$ appear consecutively in $\Tgt$,
so all vertices $w_{i,j}$ are equal (and denoted $w_i$, $1\leq i\leq \kappa$). In the edge section, each string $\sigma(X_{i,j}X_{j,i})=w_iw_j$ for $1\leq i < j \leq k$ is a substring of $\Tgt$ so it is an edge of $G$, i.e. $\{w_1, \ldots, w_{\kappa}\}$ is a clique in $G$. 
\end{proof}

Building on the result above, we prove that {\sc String Equations} is \Wone-hard even if blocks are allowed to be assigned empty strings (in this case, we need two distinct target strings: having a single target would have trivial solutions by assigning empty strings to all but one blocks).
\begin{theorem} \label{prop:W-hard-2eq-emptyBlocks}
{\sc String Equations} allowing empty blocks with $r=2$ duplicate-free equations and parameter $c$ is \Wone-hard.
\end{theorem}
\begin{proof}
See \Cref{fig:W-hard-2eq-emptyBlocks}.
We build on the proof of \Cref{prop:W-hard-2eq}, and create a similar instance of {\sc String Equations} with more separating gadgets, mainly to enforce that strings $\sigma(Z)$ and $\sigma(X_{i,j})$ may not be empty. We start from $\Tgt,\Src,\Src'$ as defined in the previous proof. We introduce characters $\gamma$, $\phi_1$, and $\phi_2$, as well as blocks $\Gamma_0,\Gamma_0',\Gamma_{h,i}, \Gamma'_{h,i}$ (with $h\in\{0,1\}$ and $1\leq i\leq \kappa$), $\Phi_1$ and $\Phi_2$.

We build the new equation system $\Ecal$ as follows (differences to the construction in the proof of \Cref{prop:W-hard-2eq} are highlighted):
\begin{align*} 
\Ecal&:=\{\Tgt\splitsIn \Src, \bm{ \Tgt'}\splitsIn \Src'\}, \text{ where: }  \\
\Tgt&:=
        \bm{\gamma^{2\kappa +1}  \ 
            \phi_1\phi_2}
        \quad z\ \bm{\gamma}\  x_0\ 
        \prod_{i=1}^n (\bm{\gamma}\ v_i^{k-1}\ \bm{\gamma} \ x_i )\quad       
        z\quad y_0\ \prod_{j=1}^m (e_j\ y_j)  \\    
\Tgt'&:=
        \bm{\gamma^{2\kappa +1}  \ 
            \phi_2\phi_1 }
        \quad z\ \bm{\gamma}\  x_0\ 
        \prod_{i=1}^n (\bm{\gamma}\ v_i^{k-1}\ \bm{\gamma} \ x_i )\quad        
        z\quad y_0\ \prod_{j=1}^m (e_j\ y_j)  \\
\Src&:=  \bm{\Gamma_0\prod_{i=1}^{\kappa} \Gamma_{0,i} \Gamma_{1,i} 
          \Phi_1\Phi_2 } 
\quad 
     Z\ \bm{\Gamma_0'}\ A_0\ 
     \prod_{i=1}^{\kappa} \big( \bm{\Gamma_{0,i}'}\   \prod_{j\neq i} (X_{i,j}  \bm{\Gamma_{1,i}'}  \ A_{i,j}) \big) \\
     &&
     \mathllap{Z'\ B_0\ \prod_{1\leq i < \kappa}\prod_{i< j \leq \kappa} (  X'_{i,j}X'_{j,i}\ B_{i,j} )} \\
\Src'&:=
\underbrace{\bm{\Gamma'_0\prod_{i=1}^{\kappa} \Gamma'_{0,i} \Gamma'_{1,i} 
          \Phi_2\Phi_1 }}_{\text{prefix section}} 
\quad 
     Z'\ \bm{\Gamma_0}\ A_0\  \prod_{i=1}^{\kappa} \big( \bm{\Gamma_{0,i}}  \prod_{j\neq i}(X'_{i,j}\  \bm{\Gamma_{1,i}}\ A_{i,j}) \big) 
     \\
     &&\mathllap{Z\ B_0\  \prod_{1\leq i < \kappa}\prod_{i< j \leq \kappa} ( X_{i,j}X_{j,i}\ B_{i,j} )} \\ 
\end{align*}

\begin{figure}
\begin{align*}
\Tgt &=\gamma\gamma\gamma\gamma\gamma\gamma\gamma\phi_1\phi_2 z \gamma\ x_0\ \gamma \ur a\ur a\gamma \ x_1\ \gamma bb\gamma \ x_2\ \gamma \ug c\ug c\gamma \ x_3\ \gamma \ub d\ub d\gamma \ x_4\ldots
\\
\Tgt' &=\gamma\gamma\gamma\gamma\gamma\gamma\gamma\phi_2\phi_1 z \gamma\ x_0\ \gamma \ur a\ur a\gamma \ x_1\ \gamma bb\gamma \ x_2\ \gamma \ug c\ug c\gamma \ x_3\ \gamma \ub d\ub d\gamma \ x_4\ldots
\\
\Tgt&\splitsIn \Gamma_{0}\Gamma_{0,1}\ldots\Gamma_{1,3}\Phi_1\Phi_2\ Z\ \Gamma'_{0} A_0\ \Gamma'_{0,1} \ur{R_g}\ur{R_b}\Gamma'_{1,1} \ A_1\  \Gamma'_{0,2}\ug{G_r}\ug{G_b}\Gamma'_{1,2} \ A_2\ \Gamma'_{0,3} \ub{B_r}\ub{B_g}\Gamma'_{1,3} \ A_3\ldots
\\
\Tgt&\splitsIn \Gamma'_{0}\Gamma'_{0,1}\ldots\Gamma'_{1,3}\Phi_2\Phi_1\ Z'\ \Gamma_{0} A_0\ \Gamma_{0,1} \ur{R_g'}\ur{R_b'}\Gamma_{1,1} \ A_1\  \Gamma_{0,2}\ug{G_r'}\ug{G_b'}\Gamma_{1,2} \ A_2\  \Gamma_{0,3}\ub{B_r'}\ub{B_g'}\Gamma_{1,3} \ A_3\ldots
\end{align*}
\caption{\label{fig:W-hard-2eq-emptyBlocks} Prefix and vertex sections in the reduction of \Cref{prop:W-hard-2eq-emptyBlocks} (with $r=2$ equations and empty blocks) applied to the graph from \Cref{fig:graph}}
\end{figure}
This concludes the construction of the instance. We now show the correctness of the reduction.

($\Rightarrow$) Assume that $G$ has a $\kappa$-clique. We set $\gamma=\sigma(\Gamma_0)=\sigma(\Gamma_0')=\sigma(\Gamma_{h,i})=\sigma(\Gamma_{h,i}')$ for $h\in\{0,1\}$ and $1\leq i\leq \kappa$. We set $\phi_i=\sigma(\Phi_i)$, $i\in \{1,2\}$, and keep the assignments to other blocks from \Cref{prop:W-hard-2eq} (inserting character $\gamma$ within separators $G_{i,j}$ and $G'_{i,j}$ since only those around blocks $X_{i,j}$ are covered by $\Gamma_{h,i}$ and $\Gamma'_{h,i}$). The new prefix section is  correctly covered by these new blocks, as are the occurrences of~$\gamma$ in the vertex sections (those around selected vertices are covered by  $\Gamma_{h,i}$ and $\Gamma'_{h,i}$, others are inserted in separators $G_{i,j}$ and $G'_{i,j}$).

($\Leftarrow$) If $\Ecal$ has an assignment, possibly with empty blocks, we focus on proving that blocks $Z$ and $X_{i,j}$ are not empty. First consider characters $\phi_i$: they must each be in their own block ($\Tgt$ and $\Tgt'$ have no other common substring containing these). Write $A$ and $B$ respectively for the blocks assigned $\phi_1$ and $\phi_2$, then $AB$ is a substring of $\Src$ and $BA$ is a substring of $\Src'$. It can be verified that $\Phi_1$, $\Phi_2$ are the only blocks satisfying this property, so $\sigma(\Phi_i)=\phi_i$ for $i=1,2$. Thus, from the prefix section of $\Tgt$ and $\Src$, we have $\sigma(\Gamma_0\prod_{i=1}^{\kappa} (\Gamma_{0,i} \Gamma_{1,i})) = \gamma^{2\kappa+1}$, so each $\sigma(\Gamma_0)$, $\sigma(\Gamma_{h,i})$ only contains characters $\gamma$, and may not contain more than one (since there is no $\gamma\gamma$ outside the vertex section, where the corresponding block is matched in $\Tgt'\splitsIn \Src'$), so it must contain exactly $\gamma$. Similarly, $\sigma(\Gamma'_0)=\sigma(\Gamma'_{h,i})=\gamma$. 
It follows that $\sigma(Z)$ is not empty and starts with $z$, which gives the separation into vertex and edge sections as in the proof of \Cref{prop:W-hard-2eq}.  For each $i$,  string $\sigma(\prod_{j\neq i}X_{i,j})$ only contains characters from vertex and edge sections (so only of the form $v_{i'}$, not $\gamma$ or separators), and $\Tgt$ contains $\gamma \sigma(\prod_{j\neq i}X_{i,j})\gamma $ as a substring, so again $X_{i,j}=w_i$ for some $i$ and every $j\neq i$. Finally,  $w_iw_j$ must be an edge of $G$ for each $i$ and~$j$, so $\{w_1,\ldots,w_{\kappa}\}$ forms a clique of $G$.

\end{proof}

We note that the above reduction requires two distinct target strings (otherwise, allowing empty blocks leads to trivial solutions where a single block is non-empty). However, the two strings differ by a single character inversion, that force each and every other block to be non-empty.

\section{Constant Equation Sizes}\label{sec:constant_c}
We now consider the special case where the right hand side of each equation in the system has only a constant number of blocks. In other words, the case when~$c$ is constant. We show that the case where~$c=2$ can be solved in polynomial time. In fact, we show that a more general case is polynomial-time solvable. To describe this special case, we introduce the following notation.

A block $X_p$ is a \emph{border block} of an equation if it appears as first or as last block in the right hand side of that equation. A block is a border block of a system of equations~$\Ecal$ if it is a border block of at least one equation of $\Ecal$. We say that $\Ecal$ has \emph{only border blocks} if non-border blocks are all jokers. In particular, in such a setting, a border block of $\Ecal$ is a border block of any equation in which it occurs (otherwise, it would both appear twice and be a joker: a contradiction).

\begin{proposition} \label{prop:P-size2} \label{prop:P-border-blocks}
{\sc String Equations} with only border blocks is polynomial-time solvable. In particular, {\sc String Equations}  with equation size $c=2$ can be solved in polynomial time.
\end{proposition}
\begin{proof}

Consider an equation~$\Tgt_i\equiv \Src_i$. We say that the string $\Tgt_i$ \emph{starts} with block $X_p$ if $\Src_i[1]=X_p$. Similarly, $\Tgt_i$ \emph{ends} with $X_p$ if $\Src_i[|\Src_i|]=X_p$).  Some length $\ell\in \mathds{N}$ is \emph{valid} for block $X_p$ if the length-$\ell$ prefixes of all strings starting with $X_p$ and the length-$\ell$ suffixes of all strings ending with $X_p$ are equal (then we denote this common substring $\sigma_\ell(X_p)$).

We let $\{X_1,\ldots,X_{k'}\}\subseteq \Blocks$ denote the border blocks of the instance. Recall that, since the instance has only border blocks, all other blocks are jokers, and border blocks are border blocks in every equation that contains them. Build a 2SAT formula as follows. For each border block $X_p$ and each integer $0\leq \ell\leq \Tgtlen$, introduce a boolean variable denoted $X_p\preceq\ell$ (which corresponds to assigning a string with length at most $\ell$  to $X_p$, its negation is denoted $X_p\succ\ell$). Create formula $\Phi_\Ecal$ using the following 1- and 2-clauses.
\begin{enumerate}
    \item \label{clause:not_empty}For each $X_p$, add clause  $X_p\succ0$.
    \item \label{clause:monotonous} For each $X_p$ and $0\leq  \ell< \ell' \leq n$, add clause  $X_p\preceq\ell \rightarrow  X_p\preceq\ell'$.  
    \item \label{clause:valid_length} For each $X_p$ and $1\leq \ell\leq n$. If $\ell$ is not a valid length for $X_p$, add clause $X_p\preceq\ell \rightarrow  X_p\preceq\ell-1$.    
    \item \label{clause:single_block} For each equation $\Tgt_i\splitsIn \Src_i$ where $\Src_i=X_p$,  add clause  $X_p\succ  |\Tgt_i| - 1$. 
    \item \label{clause:two_blocks} For each equation $\Tgt_i\splitsIn \Src_i$ where $\Src_i=X_pX_q$, and  $0\leq \ell\leq|\Tgt_i|$, add clause  $X_p\preceq \ell \rightarrow  X_q\succ |\Tgt_i| - 1 -\ell$.
  
    \item \label{clause:small_blocks} For each equation $\Tgt_i\splitsIn \Src_i$ where $\Src_i$ starts with $X_p$ and ends with $X_q$, for each $0\leq \ell\leq|\Tgt_i|$, add clause
    $X_p\succ\ell \rightarrow  X_q\preceq |\Tgt_i| - |\Src_i| -\ell+1$.
\end{enumerate}
We then prove the following claim, which completes the proof that $\Ecal$ can be solved in polynomial time:
$$\Ecal \text{ admits a satisfying assignment} \Leftrightarrow \Phi_\Ecal\text{ is satisfiable.} $$

($\Rightarrow$) Consider a satisfying assignment $\sigma$ of $\Ecal$, let $\ell_p$ be the length of $\sigma(X_p)$, and set all variables $X_p\preceq \ell$ to true iff $\ell\geq \ell_p$. We now verify each clause (each time we use the notations from the clause definition, and for clauses of the form $A\rightarrow B$, we assume that $A$ is true and directly prove $B$).
\begin{enumerate}
    \item[\ref{clause:not_empty}.] $\ell_p>0$ so $X_p\succ0$.
    \item[\ref{clause:monotonous}.] We have $\ell_p\leq \ell$ and $\ell<\ell'$ so $\ell_p\leq \ell'$.
    \item[\ref{clause:valid_length}.] $\ell_p$ is a valid length, so $\ell\neq \ell_p$, if $X_p\preceq \ell$, then $\ell_p < \ell$ and $X_p\preceq \ell-1$. 
    \item[\ref{clause:single_block}.] We have $|\Tgt_i|=|\sigma(\Src_i)|=\ell_p$, so $\ell_p>|\Tgt_i-1|$ and $X_p\succ|\Tgt_i|-1$.
    \item[\ref{clause:two_blocks}.] We have $|\Tgt_i|=|\sigma(\Src_i)|=\ell_p+\ell_q$, and $\ell\geq \ell_p$ so $\ell_q\geq |\Tgt_i|-\ell$ and $X_q\succ|\Tgt_i|-1-\ell$.   
    \item[\ref{clause:small_blocks}.] We have $\Tgt_i=\sigma(\Src_i)$, so $|\Tgt_i|\geq \ell_p+\ell_q+|\Src_i|-2$ (blocks of $\Src_i$ have length at most 1, except the first and last that have length $\ell_p$ and $\ell_q$). 
    If $X_p\succ\ell$, then $\ell <\ell_p$, and $\ell_q\leq |\Tgt_i|-\ell_p-|\Src_i|+2 \leq |\Tgt_i|-\ell-|\Src_i|+1$.
        
\end{enumerate}

($\Leftarrow$) For each $p$, let $\ell_p$ be the smallest value of $\ell$ such that $X_p\leq \ell$ is true. In particular, $\ell_p>0$ (since $X_p\preceq 0$ is false, Clause~\ref{clause:not_empty}), and $\ell_p$ is a valid length for $X_p$ (Clause~\ref{clause:valid_length}: otherwise, we would also have $X_p\preceq \ell_p-1$ set to true).
Set $\sigma(X_p)=\sigma_{\ell_p}(X_p)$. Consider an equation $\Tgt_i\splitsIn \Src_i$. By Clause~\ref{clause:valid_length}, if $\Src_i$ starts with $X_p$ and ends with $X_q$, then $\sigma(X_p)$ is a prefix of $\Tgt_i$ and $\sigma(X_q)$ is a suffix of $\Tgt_i$. In particular, $\ell_p,\ell_q\leq |T_i|$.
If $|\Src_i|=1$ ($\Src_i=X_p$), then by Clause~\ref{clause:single_block}, $\ell_p=|\Tgt_i|$, so $\sigma(X_p)=\Tgt_i$. If $\Src_i=X_pX_q$, then by Clauses~\ref{clause:small_blocks} and~\ref{clause:two_blocks} $\ell_p+\ell_q=|\Tgt_i|$, so $\sigma(X_pX_q)=\Tgt_i$. Otherwise ($\Src_i$ contains $|\Src_i|-2\geq 1$ joker blocks),  we have $\ell_p+\ell_q+|\Src_i|-2\leq |\Tgt_i|$ (indeed,  $X_p\succ\ell_p-1$, so $X_q\preceq |\Tgt_i| - |\Src_i| -(\ell_p-1)+1$ by Clause~\ref{clause:small_blocks}, hence $\ell_q \leq |\Tgt_i| - |\Src_i| -\ell_p+2$, which gives the result).
\end{proof}

\begin{theorem}  \label{prop:W-hard-size3}
{\sc String Equations} with duplicate-free equations of size $c=3$ and parameter $r$ is W[1]-hard.
\end{theorem}
\begin{proof}

\begin{figure}
\begin{align*}
\Tgt_{1,r} &:= x \ \ur ab\ x    & \Tgt_{1,g} &:= x\ \ug c\ x & \Tgt_{1,b} &:= x\ \ub d\ x  \\
 \Tgt_{1,r}&\splitsIn \joker\ \ur R\ \joker    & 
 \Tgt_{1,g} &\splitsIn \joker\ \ug G\ \joker & 
 \Tgt_{1,b} &\splitsIn \joker\ \ub B\ \joker  \\[.5em]
  \Tgt_2 &:= \mathrlap{y\  a\,c\ y\  a\,d\ y\  b\,c\ y\  b\,d\ y\  c\,d\ y}\\
  \Tgt_2&\splitsIn  \joker\ E_{rg}\ \joker  
& \Tgt_2&\splitsIn  \joker\ E_{rb}\ \joker  
& \Tgt_2&\splitsIn  \joker\ E_{gb}\ \joker  
\\[.5em]
  \Tgt_3 &:= \mathrlap{z\  \ur a\,\ug c\ z\  \ur a\,\ub d\ z\  b\, c\ z\  b\,d\ z\  \ug c\,\ub d\ z}\\
  \Tgt_3&\splitsIn  A_{rg}\ E_{rg}\ B_{rg}
& \Tgt_3&\splitsIn  A_{rb}\ E_{rb}\ B_{rb}  
& \Tgt_3&\splitsIn  A_{gb}\ E_{gb}\ B_{gb}\\
  \Tgt_3&\splitsIn  A_{rg}\ \ur R\quad \joker
& \Tgt_3&\splitsIn  A_{rb}\ \ur R\quad   \joker
& \Tgt_3&\splitsIn  A_{gb}\ \ug G\quad \joker\\
  \Tgt_3&\splitsIn  \quad \joker\quad \ug G\ B_{rg}
& \Tgt_3&\splitsIn \quad \joker\quad \ub B\ B_{rb}  
& \Tgt_3&\splitsIn \quad \joker\quad \ub B\ B_{gb}
\end{align*}
\caption{\label{fig:W-hard-size3} Reduction of \Cref{prop:W-hard-size3} (using equations of size $c=3$) applied to the graph from \Cref{fig:graph}. 
}
\end{figure}

See \Cref{fig:W-hard-size3}. Consider an instance $(G=(V,E), \kappa)$ of {\sf Multi-Colored Clique}.  Introduce three separators  denoted $x,y,z$, let $\Sigma=V\cup\{x,y,z\}$. Introduce $\kappa+3{\kappa \choose 2}$ non-joker blocks $\{X_i\mid 1\leq i\leq \kappa\}$ and $\{E_{i,j},A_{i,j},B_{i,j} \mid 1\leq i <j\leq \kappa\}$, and $2\kappa + 4{\kappa \choose 2}$ joker blocs (for a total of $k=\Oh( \kappa^2)$ blocks). The equation system is defined as follows.

\begin{align*}
    \Ecal&:=\{\Tgt_{1,i}\splitsIn \Src_{1,i} \mid 1\leq i\leq \kappa \} \\
     &\cup \{\Tgt_2\splitsIn \Src_{2,i,j} \mid 1\leq i< j \leq \kappa \}\\
     &\cup \{\Tgt_3\splitsIn \Src_{3,i,j}, \Tgt_3\splitsIn \Src_{4,i,j},  \Tgt_3\splitsIn \Src_{5,i,j} \mid 1\leq i<j\leq \kappa \}  
\end{align*}
\begin{align*}
     \Tgt_{1,i}&:=x\prod_{v_j\in V_i} (v_i y) &
     \Src_{1,i}&:= \joker \ \ X_i \ \  \joker 
     \\
     \Tgt_2&:=y\prod_{j=1}^m (e_j y)
     &
     \Src_{2,i,j}&:= \joker  \ \  E_{i,j} \ \  \joker 
     \\     
     &&     \Src_{3,i,j}&:=  A_{i,j} E_{i,j} B_{i,j} \\
     \Tgt_{3}&:=\smash{ z\prod_{j=1}^m (e_j z)}
     &\Src_{4,i,j}&:=  A_{i,j} X_{i} \quad \joker \\
     &&\Src_{5,i,j}&:=  \quad \joker\quad   X_{j} B_{i,j} 
\end{align*}
If $G$ contains a clique, let $\sigma(X_i)$ and $\sigma(E_{i,j})$ correspond respectively to the selected vertex in component $i$ and to the selected edge between  components $i$ and $j$. In particular, $\sigma(X_i)$ appears in $\Tgt_{1,i}$ and $\sigma(E_{i,j})$ appears in $\Tgt_{2}$ and $\Tgt_3$ (using $\Src_{3,i,j}$, $\sigma(A_{i,j}E_{i,j})$ is a prefix of $\Tgt_3$ and $\sigma(E_{i,j}B_{i,j})$ is a suffix of $\Tgt_3$).
Moreover, $\sigma(E_{i,j)=\sigma(X_i X_j}$, 
so $\sigma(A_{i,j}X_{i})$ is a prefix of $\Tgt_3$ and $\sigma(X_{j}B_{i,j})$ is a suffix of $\Tgt_3$, satisfying equations for $\Src_{4,i,j}$ an $\Src_{5,i,j}$, and $\sigma$ is a satisfying assignment of $\Ecal$.

Assume now that $\Ecal$ admits a satisfying assignment $\sigma$. Note that $\sigma(X_i)$ is a common substring of $\Tgt_1$ and $\Tgt_3$, so it has length 1 and does not contain any separator (there are no common separators in $\Tgt_1$ and $\Tgt_3$ and no length-2 substring of $\Tgt_3i$ uses two vertices from the same subset $V_i$). Let $v_i=\sigma(X_i)$. Similarly $\sigma(E_{i,j})$ does not contain any separator and has length at most 2 (note that if it has length 2, then it is an edge of $G$). Using $\Src_{3,i,j}$ and $\Src_{4,i,j}$, $\sigma(E_{i,j})$  and $\sigma(X_{i})$ must start with the same character, $v_i$. Similarly, $\sigma(E_{i,j})$  and $\sigma(X_{j})$ must end with the same character $v_j$. Since$v_i\in V_i$, $v_j\in V_j$ and $V_i,V_j$ are disjoint for $i\neq j$, we have $\sigma(E_{i,j})=v_iv_j$, i.e. there exist an edge $\{v_i,v_j\}$ in $G$, and $\{v_1,\ldots,v_{\kappa}\}$ form a clique of $G$.
\end{proof}

\begin{theorem} \label{prop:W-hard-all-but-one-size2}
{\sc String Equations} with $r-1$ size-$2$ equations and one size-$c$ duplicate-free equation is \Wone-hard for  parameter $r+c$.
\end{theorem}
\begin{proof}

\begin{figure}
\begin{align*}
\Tgt_{v} &:=\ugd{\ur x\,\ a\,b}\,\ugd{c\,\ub{d}} \\
\Tgt_v&\splitsIn  R\,R' &  
\Tgt_v&\splitsIn {G}\,{G'} &
\Tgt_v&\splitsIn B\,{B'} \\
\Tgt_v&\splitsIn \ur{R_g}\,R' & 
\Tgt_v&\splitsIn \ug{G_b}\,{G'} & 
\Tgt_v&\splitsIn B\,\ub{B_r'} \\
\Tgt_v&\splitsIn \ur{R_b}\,R' &
\Tgt_v&\splitsIn {G}\,\ug{G_r'} & 
\Tgt_v&\splitsIn B\,\ub{B_g'}  \\
\Tgt_{l} &:=\mathrlap{y\ \ur{x}\,\ug{cd}\ y\ \ur{x}\,\ub{d}\ y\ xacd\ y\ xad\ y\ \ug{xab}\,\ub{d}\ y}\\
\Tgt_l &\splitsIn \mathrlap{\joker\  \ur{R_g}\,\ug{G'_r}\ \joker \ \ur{R_b}\,\ub{B'_r}\ \joker \ \ug{G_b}\,\ub{B'_g}\ \joker}
\end{align*}
\caption{\label{fig:W-hard-all-but-one-size2} Reduction of \Cref{prop:W-hard-all-but-one-size2} (using $r-1$ size-2 equations and one size-$c$ duplicate-free equation) applied to the graph from \Cref{fig:graph}. }
\end{figure}

See \Cref{fig:W-hard-all-but-one-size2}. Consider an instance $(G=(V,E), \kappa)$ of {\sf Clique}.  Introduce additional characters $x$ and $y$, and let $\Sigma=V\cup\{x,y\}$. 
Given a vertex $v_i\in V$, define the (non-empty) strings $\leftpart(v_i):=x \prod_{i'=1}^{i-1} v_{i'}$ and $\rightpart(v_i):= \prod_{i'=i}^{n} v_{i'}$. We first prove the following property on these functions: 
\begin{claim}\label{claim:leftrightparts}
If $v_i,v_j,v_{i'},v_{j'}$ are vertices with $i<j$ and $ \leftpart(v_i)\rightpart(v_j)=\leftpart(v_{i'})\rightpart(v_{j'})$, then $i=i'$ and $j=j'$.
\end{claim}
\begin{proof}Character $v_i$ does not appear in  $ \leftpart(v_i)\rightpart(v_j)$, so $i'\leq i <j'$. Furthermore, $v_{i-1}$ (or $x$ if $i=0$) appears in $ \leftpart(v_i)\rightpart(v_j)$ but may not appear in $\rightpart(v_{j'})$, so it appears in $\leftpart(v_{i'})$ and $i = i'$. Removing the common prefix yields $\rightpart(v_j)=\rightpart(v_{j'})$, and $j=j'$.
 \end{proof}

Introduce $2\kappa+2{\kappa \choose 2}$ non-joker blocks $\{X_i\mid 1\leq i\leq \kappa\}$ and $\{X_{i,j},X'_{j,i} \mid 1\leq i <j\leq \kappa\}$, and ${\kappa \choose 2}+1$ joker blocs (for a total of $k=\Oh( \kappa^2)$ blocks). The equation system is defined as follows.

\begin{align*}
    \Ecal&\mathrlap{:=\{\Tgt_{v}\splitsIn \Src_{i} \mid 1\leq i\leq \kappa \} } \\
     &\mathrlap{\cup\  \{\Tgt_v\splitsIn \Src_{1,i,j},\Tgt_v\splitsIn \Src_{2,i,j}  \mid 1\leq i< j \leq \kappa \}}\\
     &\mathrlap{\cup\  \{\Tgt_e\splitsIn \Src_e \},       \text{ where: } }\\
     \Tgt_v&:=x\ \prod_{i=1}^{n} v_i     &  \Tgt_e&:= y\ \prod_{uv \in E} (\leftpart(u)\rightpart(v)\ y) \\
     \Src_{i}&:= X_i \ \ \ X_i' &
     \Src_{e}&:=   \smash{\joker\ \prod_{1\leq i<j\leq \kappa} (X_{i,j}X'_{j,i} } \ \joker)\\
     \Src_{1,i,j}&:= X_{i,j} \ X_i'\\
     \Src_{2,i,j}&:= X_j\ X'_{j,i}\\
\end{align*}

If $G$ contains a clique $K=\{w_1,\ldots, w_\kappa\}$, let $\sigma(X_i):=\leftpart(w_i)$ and $\sigma(X_i'):=\rightpart(w_i)$. For each $i<j$, let $\sigma(X_{i,j})=X_i$ and $\sigma(X'_{i,j})=X_j'$. Since, for each $w_i$, we have $\Tgt_v=\leftpart(w_i)\rightpart(w_i)$, the equations $\Tgt_v\splitsIn \Src_i$,  $\Tgt_v\splitsIn \Src_{1,i,j}$ and  $\Tgt_v\splitsIn \Src_{2,i,j}$ are satisfied. 
Furthermore, since $K$ is a clique, for any $i<j$ we have $\sigma(X_{i,j}X'_{j,i})=\leftpart(u)\rightpart(v)$ for some edge $uv$ of $G$, so  strings $\sigma(X_{i,j}X'_{j,i})$ are indeed substrings of $\Tgt_e$, in the same order as in $\Src_e$, so $\Tgt_e\splitsIn \Src_e$ is satisfied as well. 

Assume now that $\Ecal$ admits a satisfying assignment $\sigma$. Note that by $\Tgt_s\splitsIn \Src_i$, $\sigma(X_i)=\leftpart(w_i)$ and $\sigma(X'_i)=\rightpart(w_i)$ for some vertex $w_i\in V$. Using other equations with $\Tgt_s$, we get $\sigma(X_{i,j})=\sigma(X_i)$ and $\sigma(X'_{j,i})=\sigma(X'_j)$ for each $i<j$. For each $i<j$, write $E_{i,j}=\sigma(X_{i,j}X'_{j,i})$. Then $E_{i,j}$ is a substring of $\Tgt_e$ starting with $x$, ending with $v_n$, and not containing any occurrence of $y$, so $E_{i,j}=\leftpart(u)\rightpart(v)$ for some edge $uv$ of $G$. By Claim~\ref{claim:leftrightparts}, since $E_{i,j}=\leftpart(w_i)\rightpart(w_j)$, we have $u=w_i$ and $v=w_j$, so $w_iw_j$ is an edge for each $i<j$: $\{w_1,\ldots,w_{\kappa}\}$ is a clique of $G$.
\end{proof}

\section{Open Questions} 

We identified and focused on parameters $k$, $r$, $c$ and $d$ for our multivariate study of the {\sc String Equations} and {\sc String Equations with Deletions} problems. Two questions remain open, both in the deletions setting. Can the \XP-algorithm for \textsc{String Equations with Deletions} parameterized by~$k+d$ (\Cref{prop:del+k-XP}) be improved to an algorithm with running time $f(d)\cdot \Tgtlen^{f(k)}$, that is, to an \FPT-algorithm for~$d$ when~$k$ is constant? Similarly, is there an $f(d) \Tgtlen^{\Oh(1)}$-time algorithm for \textsc{String Equations with Deletions} with size-2 equations (or, even better, equations with only border blocks of arbitrary size)? 

We reserve the study of other parameters for future works: possible candidates could be the alphabet size $|\Sigma|$, the size of the target string $\Tgtlen$, or the maximum length of a string assigned to a block $|\sigma(X_i)|$. However, most combinations bring hardness for the related {\sc String Morphism} problem~\cite{FSV16}. 

Finally, another variant would allow undefined characters (question marks) in the target strings. Does this introduce an additional difficulty, and can our algorithms be extended to this setting?

\bibliography{string-equations}

\end{document}